\documentclass[smallextended]{svjour3}       
\smartqed  
\usepackage{graphicx}
\usepackage{epstopdf}
\usepackage{float}
\usepackage{amsfonts}
\usepackage{graphicx}
\usepackage{amsmath}
\usepackage{appendix}
\usepackage{multicol}
\usepackage{hyperref}
\hypersetup{
  colorlinks   = true, 
  urlcolor     = blue, 
  linkcolor    = blue, 
  citecolor    = blue  
}

\newcommand{\bd}{\begin{displaymath}}
\newcommand{\ed}{\end{displaymath}}
\newcommand{\be}{\begin{equation}}
\newcommand{\ee}{\end{equation}}
\newcommand{\bea}{\begin{eqnarray}}
\newcommand{\eea}{\end{eqnarray}}
\newcommand{\bda}{\begin{eqnarray*}}
\newcommand{\eda}{\end{eqnarray*}}
\newcommand{\ba}{\begin{array}}
\newcommand{\ea}{\end{array}}

\newcommand{\ra}{\rightarrow}

\newcommand{\dd}{\mbox{\rm\,d}}

\begin{document}

\title{Noisy threshold in neuronal models: connections with the noisy leaky integrate-and-fire model.}
\titlerunning{Age-structure in neuronal populations}
\author{G. Dumont        \and
        J. Henry         \and
        C.O. Tarniceriu
}

\institute{G. Dumont \at
              Physics Department, 150 Louis Pasteur
Ottawa, Ontario, Canada K1N 6N5 \\
              \email{gdumont@uottawa.ca}           
           \and
           J. Henry \at
              INRIA Bordeaux Sud Ouest, 200 avenue de la vieille tour,  33405 Talence Cedex, France\\
              \email{jacques.henry@inria.fr}
           \and
           C.O. Tarniceriu \at
           Department of Sciences of "Al. I. Cuza" University,  %
          Lascar Catargi 54, 700107 Ia\c{s}i, Romania\\
          \email{tarniceriuoana@yahoo.co.uk}
}

\date{}

\maketitle
\begin{abstract}
Providing an analytical treatment to the stochastic feature of neurons' dynamics is one of the current biggest challenges in mathematical biology.
The noisy leaky integrate-and-fire model and its associated Fokker-Planck equation are probably the most popular way to deal with
neural variability. Another well-known formalism is the escape-rate model: a model giving the probability that a neuron fires at a
certain time knowing the time elapsed since its last action potential.
This model leads to a so-called age-structured system,
a partial differential equation with non-local boundary condition famous in the field of population dynamics,
where the {\it age} of a neuron is the amount of time
passed by since its previous spike. In this theoretical paper, we investigate the mathematical connection between the two formalisms.
We shall derive an integral transform of the solution to the age-structured model into the solution of the Fokker-Planck equation.
This integral transform highlights the link between the two stochastic processes. As far as we know, an explicit mathematical
correspondence between the two solutions has not been introduced until now.
\end{abstract}

\section{Introduction}

Neurons are strongly noisy. They never respond in the same way under repeated exposure to identical stimuli and it is difficult for
theoreticians  to apply the correct analytical treatment in order to express this variability. Two distinct sources of noise are usually mentioned:
external and internal \cite{noiseNS}. While the external source of noise usually refers to the random fluctuations attributed to
the environment of the neurons, the internal source is mainly imputed to the probabilistic nature of the chemical reactions governing
the firing process of neurons. More precisely, noise is present because a neuron is bombarded by thousands of synaptic inputs, and also
due to the randomness in the openings and closings of the ion channels underlying action potentials \cite{andre01}.

The noisy leaky integrate-and-fire (NLIF) model is a mathematical model that takes into account the stochastic features of
neurons \cite{Burkitt}.
The model is preferred by theoreticians since it can be seen as a simplification of the bio-physiological
Hodgkin-Huxley model \cite{HH}, which is sufficiently detailed to allow a qualitative comparison with physical data obtained via intracranial
recording \cite{Izi} (see also \cite{naud} for a recent discussion about the quality of the neural modeling). Nonetheless, despite
its apparent simplicity, many  questions regarding its dynamics remain open.

By definition, the NLIF model describes a stochastic process, which is given by a Langevin equation plus a discontinuous reset mechanism
to mimic the onset of the action potential (see \cite{Burkitt} and \cite{Izi}).
Starting with the Langevin equation, one can write the well-known associated Fokker-Planck (FP) equation \cite{Gardiner},
that gives the evolution in time of the
density probability to find a neuron's membrane potential in a certain voltage value \cite{longtin01}.

Let us remind that, in mathematical neuroscience, the concept of probability density function has already a long history,
as it can be seen in \cite{rinzel}, \cite{abbott}, and it is used in a variety of contexts. Indeed, assuming the number of neurons
to be infinitely large, one can write the so-called thermodynamics' mean field equation, where the effect of the whole  network on
any given neuron is approximated by a single averaged effect. Under some assumptions and approximations, the equation takes the form
of a nonlinear FP equation. It is in particular pertinent for the simulation of large sparsely connected populations of neurons
\cite{sirovich}, \cite{nykamp}, \cite{us}. Furthermore, this density approach has brought an important added value on the theoretical
understanding of synchronization and brain rhythms.  Particularly, this approach has  been successfully used to understand synchronization
caused by recurrent excitation \cite{DH}, \cite{DH01}, \cite{carillo01},
 by  delayed inhibition feedback \cite{BH01}, by both recurrent excitation and inhibition \cite{B02} and
 by gap junction \cite{brunel01}. On a similar trend, it has been used to study the occurrence of the neural cascade
 \cite{cascade01}, \cite{cascade02} and the
emergence of self criticality \cite{critical} with synaptic adaptation.

In this paper, we do not investigate the effect of interactions among neurons, but focus on the analytical treatment of neural noise.
NLIF model is a popular way to deal with the stochastic aspect of neurons, another way is the escape-rate model \cite{plesser},
\cite{gerstner}.
 The main difference between the two approaches consists in the treatment of noise; while in the NLIF model  the noise acts
on the trajectories, the escape rate model considers deterministic trajectories and the noise is present in expressing the variability
of firings that is modeled in the form of a {\it hazard function}. Therefore, the noisy trajectories with fixed threshold
are replaced
by deterministic trajectories with noisy thresholds. In the equivalent description of the NLIF model as a FP equation
with absorbing boundary condition at the firing threshold, the noise is expressed by the diffusion term of the FP equation; it
has been shown in \cite{plesser} that, in the subthreshold regime, the integrate and fire model with stochastic
input (diffusive noise) can be mapped onto an escape rate model with a certain escape rate.
Starting from this, the equivalence of the FP equation with escape noise and a partial differential equation that describes
the evolutions of {\it refractory densities} has been shown \cite{gerstner}; the last equation
is strikingly similar to those of the well-known age-structured (AS) models and it  gives the
evolution in time of the refractory densities with respect to their refractory state, which is in fact the time elapsed
since the last firing. To underline the above mentioned similarity, we will refer to this variable in this paper as {\it age}.
Age structure in a neural context has been also discussed in \cite{perthame03}.

In our paper, we shall prove that the solution to the AS system can
be transformed via an integral transform into the solution to the FP equation associated to the NLIF model.
The kernel of the integral transform will involve in particular the
notion of {\it survivor function} \cite{g2000}, \cite{gerstner}.  In renewal theory, the hazard is known also as the
{\it age dependent death rate} and expresses the rate of decay of the survivor function \cite{cox}. The concept of
 time dependent interspike interval (ISI) distribution and corresponding survivor function has been considered later \cite{G95}.

In the neuroscience context, the survivor
function, which was introduced initially to describe the probability of a particle to reach a given target, will give the
probability for a neuron to "survive" without firing. We refer again for more about these considerations to \cite{gerstner}, and further analysis on these functions
and the related first passage time problem in the neural context can be found in the review \cite{BN}. First passage time problem
in cellular domains has been investigated in \cite{schuss} and \cite{holcman}.

There is a strong advantage in using an AS formalism: the AS systems have been thoroughly investigated
in the last decades, and many qualitative results of the various forms of AS population models have been obtained.
By proving an equivalence between a membrane  potential density model and an AS model, we will be in position to obtain insights
of the qualitative behavior of the population density function such as long time behavior, stability, bifurcation points, and so on.

The paper is structured as follows: we remind in the first two sections the NLIF model and its
associated FP equation, and we present
some simulations of the models. Next, some considerations about the survivor function, the interspike interval distribution and
the first passage time
problem are presented. We introduce in the following the stochastic threshold model and
the corresponding AS system. We prove our main theoretical results in the last two sections: first we
establish in Proposition \ref{ITransformation} an integral correspondence between
its solution and the solution to FP equation. We also consider the stationary case, and show that,
by our integral transform, we obtain an expression of the corresponding
stationary solution which verifies the stationary
 FP equation.  Last, we show the asymptotic convergence of the solution to FP system to the
 solution of the stationary system defined through our transform.
The existence of an inverse transform
to the one introduced here  as well as extensions of the problem to the case of
time-dependent parameters of the systems are subject to our further investigations .

 Before getting started, let us summarize in Table
\ref{TabNotation} the main mathematical notations and their associated biophysiological meaning used throughout this document.

\begin{table}[h]
\begin{center}
\begin{tabular}{|c||l|}
  \hline
  \bf  Notation  &  \bf Biophysiological interpretation \\
  \hline
  $v(t)$        & Neuron's membrane potential  \\
  $v_r$       & Reset potential \\
  $\mu$       & Bias current \\
  $\sigma$       & Noise intensity \\
  $p(t,v)$    & Population density with respect to potential \\
  $r(t)$    & Neuron's firing rate \\
  $a(t)$         & Neuron's {\it age}, i.e. time elapsed since the last spike \\
  $q(a,v)$   & Joint probability density of the membrane's potential and neuron's {\it age}\\
  $ISI(a)$        & First passage time \\
  $S(a)$        & Age dependent death rate  \\
  $n(t,a)$   & Population density for the age-structured model \\
  \hline
\end{tabular}
\end{center}
\caption {Main notations used throughout this paper and their biophysiological interpretations}\label{TabNotation}
\end{table}

\section{The noisy leaky integrate and fire model}
The NLIF model is a well known model in the field of computational neuroscience \cite{Izi}.
The model consists in an ordinary differential
equation describing the subthreshold dynamics of a single neuron membrane's potential and the onset of an action potential described
by a reset mechanism: a spike
occurs whenever a given threshold $V_{T}$ is reached by the membrane potential variable $V$. Whenever the firing threshold is reached,
it is considered that a spike has been fired and the membrane potential is instantaneously reset to a given value $V_R$.
The dynamics of the subthreshold potentials are given by

\begin{equation*}
\tau \frac{d}{dt}V(t)=-g (V(t)-V_L)+\eta (t), \\
\end{equation*}%
where $V(t)$ is the membrane potential at time $t$, $\tau$ is the membrane capacitance, $g$ - the leak conductance,
$V_L$ - the reversal potential and $\eta (t)$ - a gaussian white noise, see \cite{B01} and \cite{abbott01} for the history of the model,
\cite{Burkitt} for a recent review and see \cite{Izi} for other spiking models. In what follows, we will use a normalized
version of the above equation, i.e.
we define $\mu$ as the bias current and  $v$ the membrane's potential which will be given by

\begin{equation*}
\mu=   \frac{V_L}{V_T}, \quad v=\dfrac{V}{V_T}, \quad v_r=\frac{V_R}{V_T}.
\end{equation*}%
After rescaling the time in units of the membrane constant $g/ \tau $, the normalized model reads

\begin{equation}
\label{IF}
\left\{
\begin{array}{l}%
\frac{d}{dt}v(t)=\mu-v(t)+\xi(t) \\
\text{If} \quad v>1 \quad \text{then} \quad v= v_r.
\end{array}%
\right.
\end{equation}%
Again, $\xi(t)$ is a Gaussian white noise stochastic process with intensity $\sigma$:

\begin{equation*}
 \langle \xi(t) \rangle =0, \quad  \langle \xi(t)\xi(t')  \rangle =\sigma \delta(t-t').
\end{equation*}%

In Fig. \ref{FigIF}, a simulation of the neuron model (\ref{IF}) is presented. The three panels correspond to
the same simulation with different level of noise. As expected, when the stochastic coefficient is increased, the corresponding
dynamics become much more irregular. Note that for $\mu$ small enough, the equilibrium of the membrane potential will be located
under the threshold.
In this situation, the neuron will fire only due to the stochastic Brownian motion.
We refer to this situation as to a subthreshold regime. In a real-world setting, such situation appears in a balanced neural network
for instance, when the excitatory and inhibitory pre-synaptic inputs cancel out.

\begin{figure}[t]
\begin{center}
    \includegraphics[width=12.5cm]{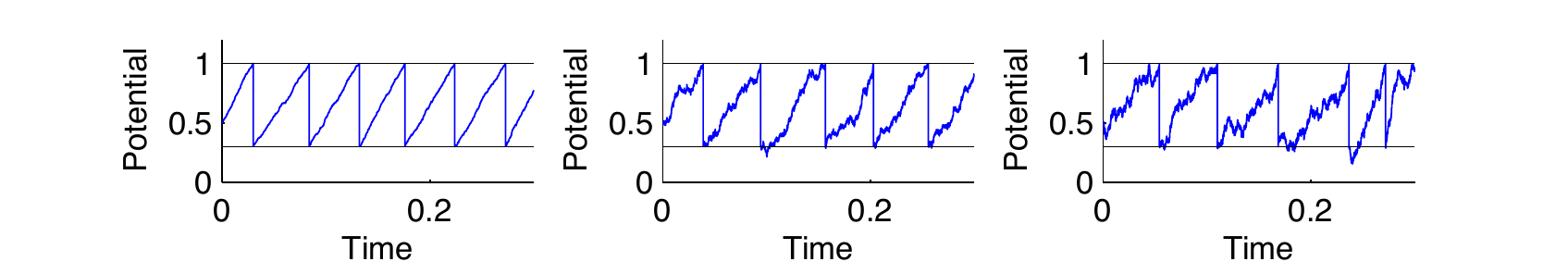}
   \caption{Simulation of the neuron model (\ref{IF}) for different values of the noise coefficient.
   The parameters of the simulation are: $v_r= 0.3$, $\mu=20$, and $\sigma=  0.1$ for the first simulation,  $\sigma=  0.4$
   for the second simulation, $\sigma=  0.6$ for the third simulation. }\label{FigIF}
      \end{center}
\end{figure}

\section{The population density function (Fokker-Planck formalism)}

Considering a population of neurons that are individually described by the stochastic equation (\ref{IF}),
the evolution of the population density function
 has been proven to satisfy the FP equation.
 We remind that the FP equation has been used in two different contexts in mathematical neuroscience:
 to model the evolution of both probability density function and
population density function.
For more considerations about the link between the two approaches we
refer to \cite{nykamp}, \cite{knight}, \cite{K2000}, \cite{BH01}.

In this paper we shall use both formalisms: we shall consider a density of neurons characterized by a population
density function, denoted here by $p(t,v)$, which satisfies the FP equation,
and each neuron of the given population has the evolution
of the potential of the membrane given by the NLIF model. Then, the probability density function of each neuron to be at a certain
voltage at a given time  will be described by the same FP equation \cite{Gardiner},
this time considered only in an inter-spike interval, as we shall see in the next section.

 This equation  is a conservation law taking into account three phenomena modeled by: a drift term due to the continuous
 evolution in the NLIF model, a diffusion term due to the noise and a term due to the reset to $v_r$ right after the firing process.
 Let $r(t)$ be the firing rate of the population, i.e. the flux through the threshold.
Then, the dynamics of the population density $p(t,v)$ is:

\begin{equation}\label{model}
\frac{\partial }{\partial t}p(t,v)+\overbrace{\frac{\partial }{\partial v}[(\mu -v)p(t,v)]}^{\text{Drift}}
-\overbrace{ \frac{ \sigma ^2}{2}     \frac{\partial^2 }{\partial v^2}p(t,v) }^{\text{Diffusion}}  =
\overbrace{\delta (v-v_{r}) r(t)}^{\text{Reset}}.\quad
\end{equation}
We show in Fig. \ref{State_Potential} a schematic representation of the state space for the FP equation (\ref{model}).
\begin{figure}[t]
\begin{center}
    \includegraphics[width=6cm]{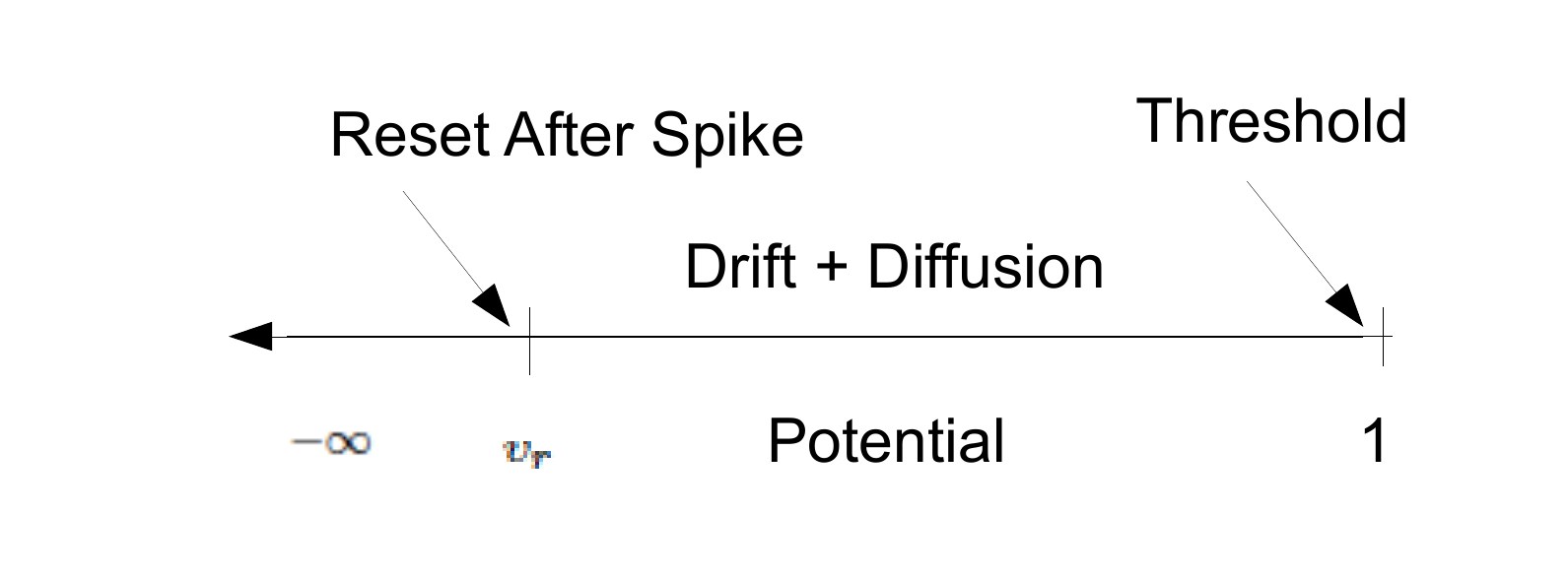}
  \caption{Schematic representation of the state space  for the FP equation (\ref{model}). }\label{State_Potential}
      \end{center}
\end{figure}
Because a neuron reaching the threshold fires an action potential and is instantaneously reset to $v_r$, we impose an absorbing
condition at the threshold (\cite{gillespie}), namely

\begin{equation}\label{BC1}
p(t,1)=0, \forall t\geq 0.
\end{equation}
Usually, a reflecting boundary is  imposed at $v=-\infty$ in order to assure the conservation property

\begin{equation}\label{BC2}
\lim\limits_{v \to -\infty} (-\mu +v)p(t,v) +\frac{ \sigma ^2}{2}     \frac{\partial }{\partial v}p(t,v) =0, \forall t\geq 0 .
\end{equation}
Of course, an initial distribution of the membrane potential is taken as a given function:
\be\label{IC}
p(0,v)=p_0(v), \quad v\in(-\infty,1].
\ee
As previously said, the firing rate $r(t)$ is defined as the flux at
 the threshold and, due to the boundary condition for the population density function in this value, is given by

\begin{equation}
\label{firingP}
 r(t)  = -\frac{ \sigma ^2}{2}     \frac{\partial }{\partial v}p(t,1).
\end{equation}
Using the boundary condition and the expression of $r(t)$ given by (\ref{firingP}), one can easily check the conservation
property of the equation (\ref{model}) by  directly integrating it on the interval $(-\infty,1)$, so that,
if the initial condition  satisfies
\begin{equation}
\label{normalisation_ini}
\int  _{-\infty}^{1} p_0(v)\,dv=1,
\end{equation}
then the solution to (\ref{model})-(\ref{firingP}) necessarily satisfies the normalization condition

\begin{equation}
\label{normalisation}
\int  _{-\infty}^{1} p(t,w)\,dw=1.
\end{equation}
Despite its "weird" singular source term, the existence of a solution to the above model has been proved in \cite{carrillo02}.
The FP equation can be written as a Stephan problem and an implicit solution can be given in the form of an integral equation.
 Note that in the literature, the equation (\ref{model}) is often exposed in terms of a conservation law.
In this setting, the flux that we denote $ \mathfrak{J}(t,v)$ is defined as

\begin{equation*}
-\mathfrak{J}(t,v)= (-\mu +v)p(t,v) +\frac{ \sigma ^2}{2}     \frac{\partial }{\partial v}p(t,v).
\end{equation*}
Therefore, the evolution in time of the density function $p$ is given by
\begin{equation*}
\frac{\partial}{\partial t}p(t,v)=-\frac{\partial}{\partial v} \mathfrak{J}(t,v).
\end{equation*}
In this formulation, the singular source term that appears in (\ref{model}) can be seen as a flux discontinuity,
see \cite{B02} for instance,

\begin{equation*}
\lim\limits_{v \to v_r^+}  \mathfrak{J}(t,v) -  \lim\limits_{v \to v_r^-}  \mathfrak{J}(t,v) = \mathfrak{J}(t,1).
\end{equation*}

\begin{figure}[t]
\begin{center}
    \includegraphics[width=12.5cm]{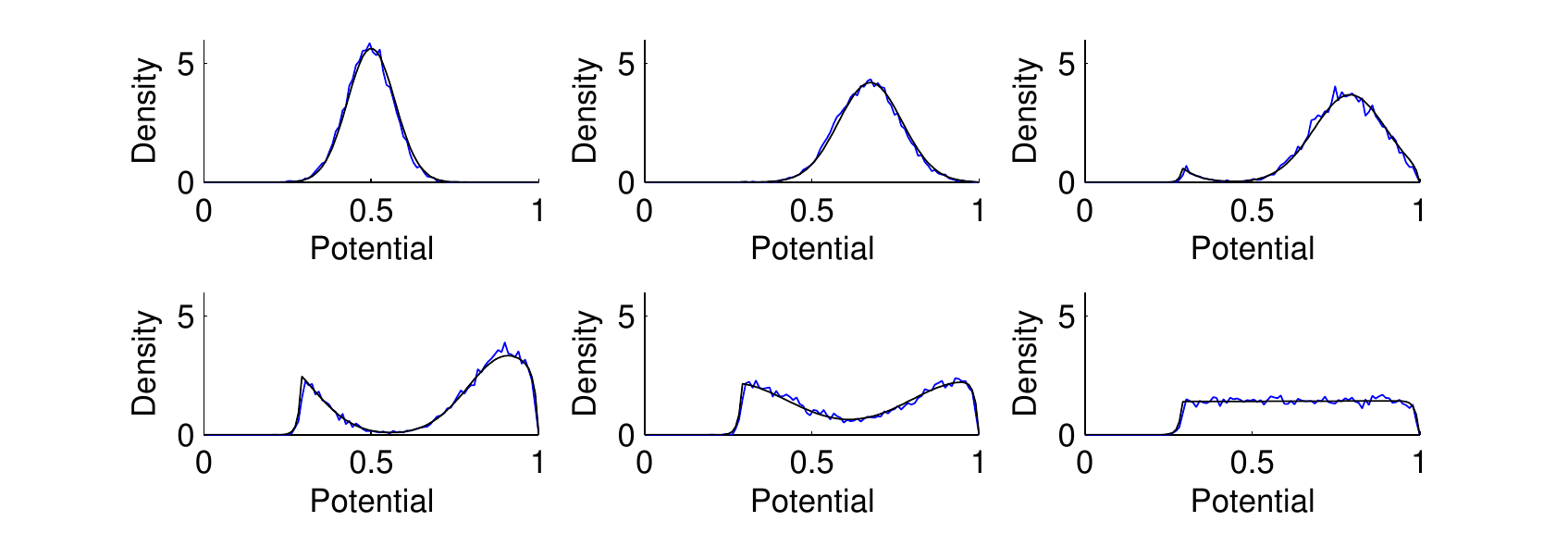}
   \caption{Simulations of the FP equation (\ref{model})-(\ref{firingP}) and of the stochastic process (\ref{IF}):
   black curve for the FP equation,
   blue curve for the stochastic process. A gaussian was taken as initial condition; the parameters of the simulation are:
   $v_r=0.3$, $\mu=20$, $\sigma =0.4$. The plots show the evolution in time of the solution at $t=0$, $t=0.1$, $t=0.3$, $t=0.5$,
    $t=0.7$, $t=7$.}\label{Fig1}
      \end{center}
\end{figure}

We present in  Fig \ref{Fig1} a simulation of the FP model (\ref{model})-(\ref{firingP}). The numerical results are compared
with  Monte Carlo simulations for the stochastic NLIF model (\ref{IF}).
In  Fig \ref{Fig1},  the black curve corresponds to the FP equation (\ref{model})-(\ref{firingP}) and the blue curve to
the stochastic process (\ref{IF}). A Gaussian was taken as  initial condition (see the first panel of Fig. \ref{Fig1}).
Under the drift and the diffusion effects, the density function gives a non zero flux at the threshold. This flux is reset to $v_r$
according to the reset process. This effect can be seen clearly in the third panel of the simulation presented in Fig. \ref{Fig1}. Asymptotically the solution  reaches a stationary density. The steady state is shown in the last panel of
Fig. \ref{Fig1}. Note that the stationary state can be easily computed (we remind its expression later in the text).
One can show the convergence of the solution towards the stationary density using the general relative entropy principle.

In Fig. \ref{Fig2} a comparison of the firing rate (\ref{firingP}) computed via the FP formalism
(\ref{model}) and via the stochastic model (\ref{IF}) is represented. Again the blue curve is obtained by  direct simulations
of the stochastic process
 (\ref{IF}), and the black curve corresponds to the simulations of the FP model (\ref{model})-(\ref{firingP}). To be more precise,
 we also show a raster plot depicting the spike timing of the neurons for each simulation run. In the three different simulations that we present, we have varied the drift term $\mu$.

 \begin{figure}[t]
 \begin{center}
 \includegraphics[width=12.5cm]{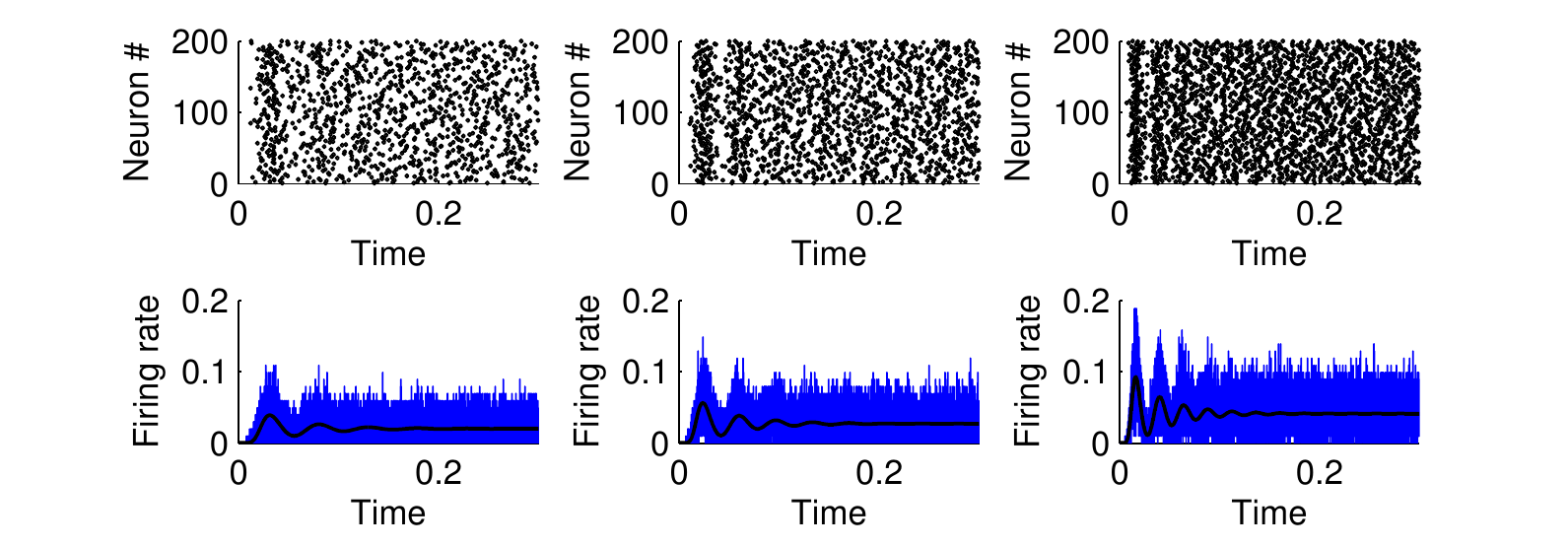}
   \caption{Simulation of the firing rate of the neuron (\ref{firingP}) via the FP formalism (\ref{model})-(\ref{firingP}),
   black curve, and the stochastic process (\ref{IF}), blue curve. The parameters of the simulation are: $v_r=0.3$,
   $\mu = 15$ and $\sigma =0.4$ for the first simulation, $\mu = 20$ and $\sigma=0.4$ for the second simulation
   and $\mu = 30$ and $\sigma =0.4$ for
   the third simulation. We also show the raster plot depicting the spike timing of $200$ neurons.}
   \label{Fig2}
      \end{center}
\end{figure}

The stationary state of the FP equation is known from decades. A straightforward computation shows that the
steady state $p_{\infty}(v)$ is given by

\begin{equation}\label{StationaryP}
p_{\infty}(v)= \frac{2r_{\infty} }{\sigma^2}  e^{-\dfrac{(v-\mu)^2}{\sigma^2}}  \int_{\max (v,v_r)}^1    e^{\dfrac{(w-\mu)^2}{\sigma^2}}\,dw,
\end{equation}
with $r_{\infty}$ the corresponding stationary firing rate (see Fig \ref{Fig2}). The latter is determined by the normalization
condition:

\begin{equation}\label{StationaryRate}
r_{\infty}^{-1}= \frac{2 }{\sigma^2} \int_{- \infty}^1 e^{-\dfrac{(v-\mu)^2}{\sigma^2}}  \int_{\max (v,v_r)}^1
e^{\dfrac{(w-\mu)^2}{\sigma^2}}\,dw \,dv.
\end{equation}
These expressions are well-known and details can be found in \cite{Ermentrout} for example.

\section{The inter-spike interval and the first passage time}

In the following we will define the {\it age} of a neuron as the time passed since its
last firing. {\it Age} is a somehow  forced notion in this context,
but we have chosen to use it due to the similarity of the model that we will present in the next section
to those from the AS systems theory. The evolution of a probability density function for a neuron's membrane potential
to be at {\it age} $a$ in the potential value $v$, $q(a,v)$, is given by a similar  FP equation:

\begin{figure}[t]
\begin{center}
    \includegraphics[width=12.5cm]{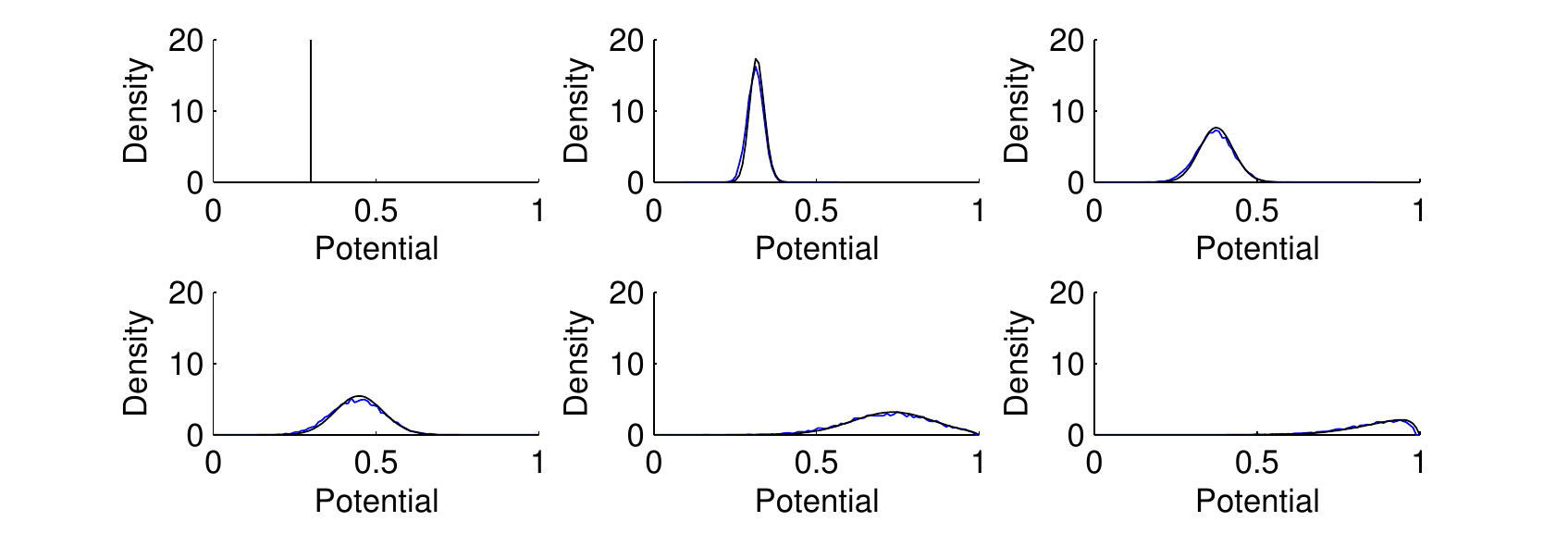}
   \caption{Simulation of the model (\ref{FP}) - (\ref{ic}), black curve, and the stochastic process (\ref{IF}), blue curve.
    A Dirac mass at the reset potential was taken as initial condition. The parameters of the simulation are $v_r=0.3$, $\mu=20$,
    $\sigma =0.4$.
    The plots show the evolution in time ({\it age}) of the solution at $a=0$, $a=0.1$, $a=0.3$, $a=0.5$, $a=0.7$, $a=7$.}\label{Fig3}
      \end{center}
\end{figure}

\begin{equation}\label{FP}
\frac{\partial }{\partial a}q (a,v)+\overbrace{\frac{\partial }{\partial v}[(\mu -v)q (a,v)]}^{\text{Drift part}}
-\overbrace{ \frac{ \sigma ^2}{2}     \frac{\partial^2 }{\partial v^2}q (a,v) }^{\text{Diffusion}}  =0,
\end{equation}
again with an absorbing boundary condition for the threshold value $v=1$
\begin{equation}\label{ba}
q(a,1)=0, \forall a\geq 0,
\end{equation}\label{br}
and a reflecting boundary condition at the boundary $v=-\infty$
\begin{equation} \label{br00}
\lim\limits_{v \to -\infty} (-\mu +v)q(a,v) +\frac{ \sigma ^2}{2}     \frac{\partial }{\partial v}q (a,v) =0, \forall a\geq 0.
\end{equation}
Since a neuron firing an action potential is reset to $v_r$, we consider the initial density given by
\be\label{ic}
q(0,v)=\delta(v-v_r), \quad v\in(-\infty,1),
\ee
where $\delta$ is the Dirac distribution.
As for the equation  (\ref{model}), the equation (\ref{FP})
is often represented in terms of an integral equation. In this setting, the flux that we denote
$ \mathfrak{F}(a,v)$ is defined as

\begin{equation*}
-\mathfrak{F}(a,v)= (-\mu +v)q(a,v) +\frac{ \sigma ^2}{2}     \frac{\partial }{\partial v}q(a,v).
\end{equation*}
Therefore, the evolution in time of the density function $q$ is given by
\begin{equation*}
\frac{\partial}{\partial a}q(a,v)=-\frac{\partial}{\partial v} \mathfrak{F}(a,v).
\end{equation*}

Note that, in the case considered here, the re-injection of the (probability) flux to the
reset value (right hand side of equation (\ref{model})) is not considered, therefore the above model represents the evolution
of the probability density of neurons before firing, i.e. in an inter-spike interval.
The interpretation is that, once the neuron
fired, it becomes of  age zero, and the source term $\delta(v-v_r)$ as initial condition can be understood intuitively
as that, right after a spike the membrane potential
is $v_r$ with probability $1$.
It should be stressed that as the bias current $\mu$ and the noise intensity $\sigma$ do not depend on time,
the probability density $q$ does not depend on time but only on  age $a$.
The flux at the threshold value, which is again given only by the diffusive part of the flux,
is a measure of great interest since it gives the ISI distribution (for a neuron having
at  age $a=0$ the potential $v_r$):

\be\label{ISI}
ISI(a)=-\frac{ \sigma ^2}{2}     \frac{\partial}{\partial v}q(a,1), \forall a\geq 0.
\ee
Until now, no analytical solution of the ISI curve has been found.

\begin{figure}[t]
\begin{center}
       \includegraphics[width=12.5cm]{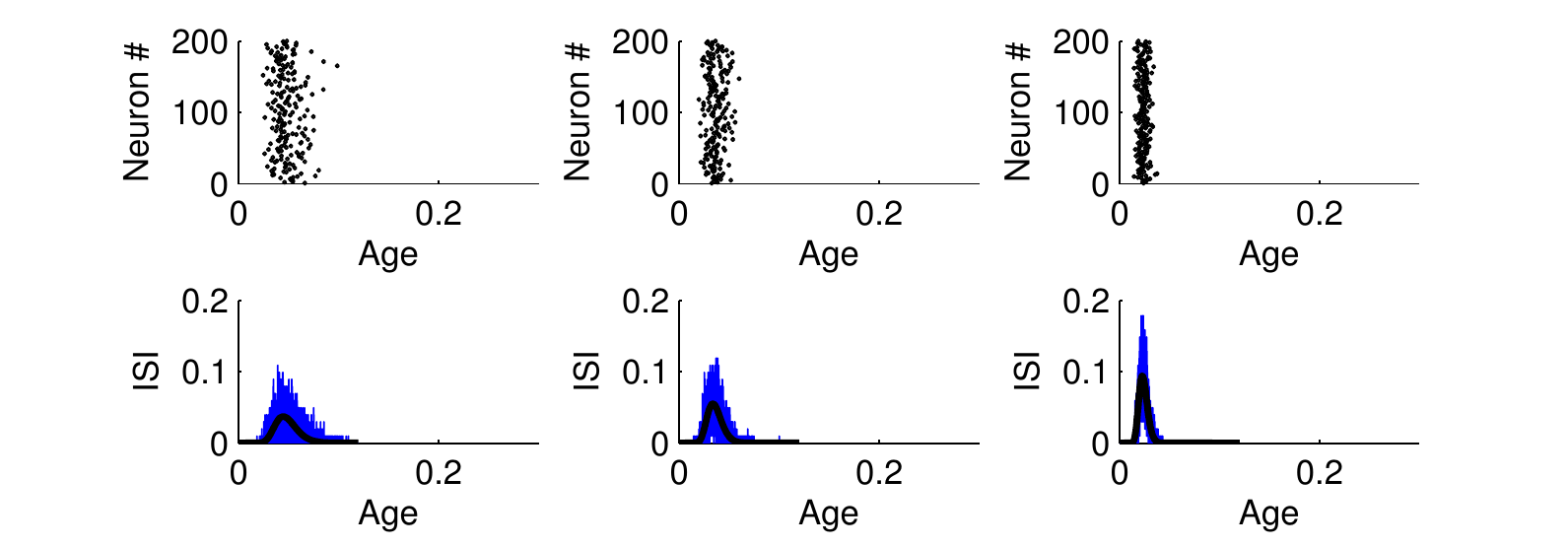}
   \caption{Simulation of the ISI  function (\ref{ISI}) via the FP formalism (\ref{FP}), black curve,
   respectively the stochastic process (\ref{IF}), blue curve. The parameters of the simulation are $v_r=0.3$,
   $\mu = 15$ and $\sigma =0.4$ for the first simulation,  $\mu = 20$ and $\sigma=0.4$ for the second simulation
   and $\mu = 30$, $\sigma =0.4$ for the third simulation.
   We also show the raster plot depicting the spike timing of $200$ neurons.}\label{Fig4}
      \end{center}
\end{figure}

We present in  Fig. \ref{Fig3} a simulation of the problem (\ref{FP})--(\ref{ic}).
Again, we have made a comparison between the stochastic
process (blue curve) and the evolution in time of the density function (black curve). The simulation starts
with a Dirac mass as initial condition (see first panel of Fig. \ref{Fig3}). Under the influence of the drift term and the
diffusion process (gaussian white noise), the density function spreads to the threshold. This is clearly seen in the upper plot
of Fig. \ref{Fig3}. At last, as expected, it converges to a zero density
(see the last panel in
Fig. \ref{Fig3}).
In  Fig. \ref{Fig4}, we make some different simulations of the ISI curve. As before, the blue curve  corresponds
to the stochastic process simulated via a Monte Carlo method, and the black curve - to the deterministic process (\ref{FP}).
We present here three different
panels corresponding to three different simulations where the bias current $\mu$ was increased. We also present
a raster plot depicting the spike timing of each neuron simulated via the NLIF model (see upper panels of Fig \ref{Fig4}).
It can be clearly seen that, increasing the intensity current leads to a more concentrated ISI density.
The ISI curve starts at zero, which means
that right after spiking the neuron needs some time before spiking again. Then, the ISI curve increases and, after reaching a
maximum, it decreases rapidly. This is  depicted by the raster plot presented in Fig \ref{Fig4}.

The first passage time problem is intimately related to the FP equation. Starting from the Chapman-Kolmogorov equation,
it has been shown that the probability to find the state (potential) of a neuron at time $t$ in a certain value v, is the solution
to the FP or Kolmogorov's forward equation. From
 it, one can derive the equation that describes the evolution in time of the probability for a neuron that started at time $0$ from
a potential value $v$ to not have reached yet the value threshold, named {\it survival probability density}. This equation
is known as Kolmogorov's backward equation, and the choice of boundary conditions has been discussed in \cite{Gardiner}.

\section{Noisy threshold model}
In what follows, we shall use the concept of survival probability or {\it survivor function} as in \cite{gerstner}. In our case,
this function will be only  age-dependent since we shall consider it for the case of a neuron that starts at age zero from
the reset potential $v_r$.

Namely, if $q (a,v)$ is the solution to (\ref{FP})-(\ref{ic}), then the following quantity

\bd
\int_{-\infty}^1 q (a,v)\dd v
\ed
stands for the  probability of survival at  age $a$ of a neuron that started at  age $0$ from the position $v_r$.
Again, "survival" at  age $a$ means that up to that time, the potential of the neuron's membrane has not reached yet
the threshold value. Then, the rate of decay of the survivor function

\begin{equation}
\label{SF}
S(a)=-\dfrac{\frac{d }{d a}  \int_{-\infty}^{1}q(a,w) \,dw}{\int_{-\infty}^{1}q(a,w) \,dw}
\end{equation}%
represents the rate at which the threshold is reached and it has been called {\it age-dependent death rate} or {\it hazard}.
$S(a)$ has the interpretation that, in order to emit a spike, the neuron has to "survive" without firing in the
interval $(0, a)$ and then fire at {\it age} $a$.

\begin{figure}[t]
\begin{center}
    \includegraphics[width=12.5cm]{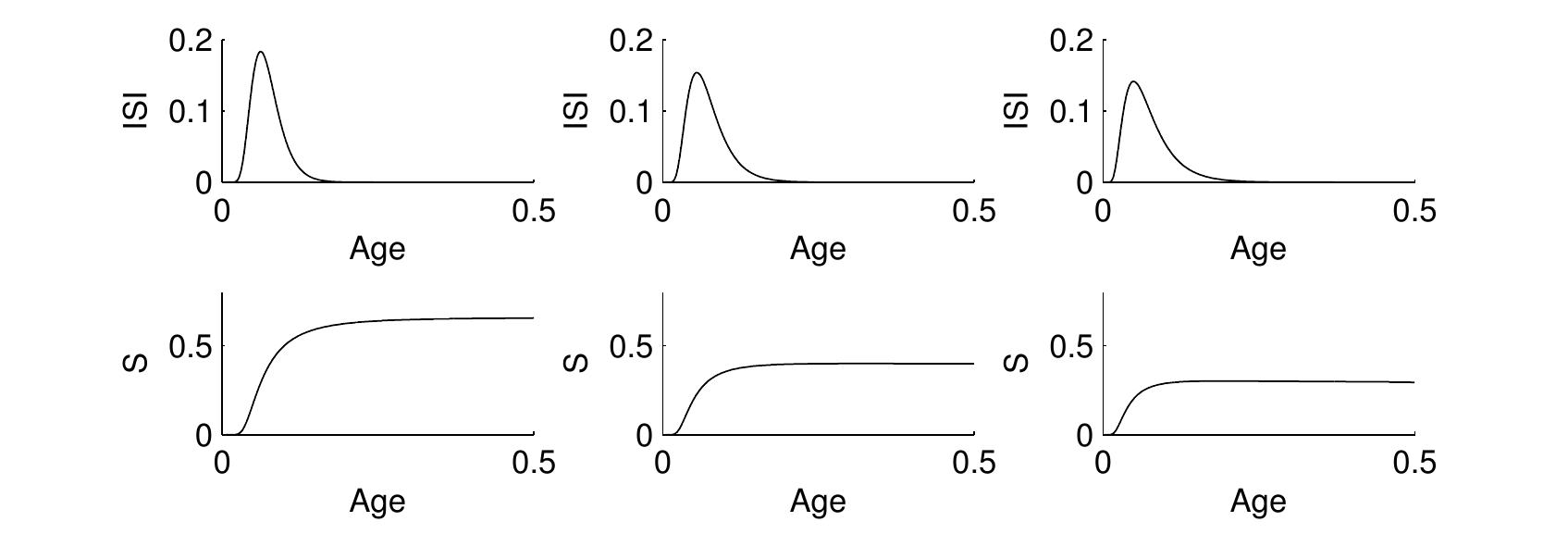}
   \caption{Simulation of the function $S(a)$ given by (\ref{SF}) in the lower panels and its corresponding
   $ISI$ given by (\ref{ISI}) in the upper panels.
    The parameters of the simulation are $v_r=0.7$, $\mu=5$, $\sigma =0.1$, $\sigma =0.2$, $\sigma =0.3$.}\label{Fig5}
      \end{center}
\end{figure}

In  Fig. \ref{Fig5}, numerical simulations of the  age-dependent death rate for different parameters is presented.
Let us notice that $S$
defines clearly a positive function that converges toward a constant. Indeed, $q(a,v)$ has the same asymptotic behavior as
\begin{equation*}
e^{-\lambda a}q(v),
\end{equation*}%
which implies that
\begin{equation*}
\lim\limits_{a \to +\infty } S(a)=\lambda,
\end{equation*}%
with $\lambda$ the dominant eigenvalue of the operator of the stationary FP equation and $q$ the
corresponding eigenvector (see \cite{ost} for example).

Note that $S$ can also be expressed in terms of the $ISI$ function given above
\begin{equation*}
S(a)=-\dfrac{ISI(a)}{1-\int_0^a ISI(s) \, da}
\end{equation*}%
which is  the expression that we used in our numerical estimations of $S$.

We can now define properly the new stochastic process. The model is given by the evolution of the  age of the neuron plus a
stochastic reset mechanism to take into account the initiation of an action potential, and it is

\begin{equation}
\label{TEM}
\left\{
\begin{array}{l}%
\displaystyle\frac{d}{dt}a(t)=1 \\
\text{The spiking probability in $(t,t+dt)$ is given by $S(a(t))dt$}\\
\text{If a spike is triggered then $a(t)= 0$},
\end{array}%
\right.
\end{equation}%
where $S$ is the  age-dependent death rate given by (\ref{SF}). In this model, the  age of the neuron follows
a trivial deterministic process, but the firing threshold is stochastic since at each time the neuron can fire.
When this happens, its  age is reset to zero.
As reminded before, the difference between the models (\ref{IF}) and (\ref{TEM}) is that,
in the NLIF the dynamics are stochastic and the reset process is deterministic while in the escape-rate model  above,
the dynamics are deterministic but the reset mechanism is stochastic.

\begin{figure}[t]
\begin{center}
    \includegraphics[width=12.5cm]{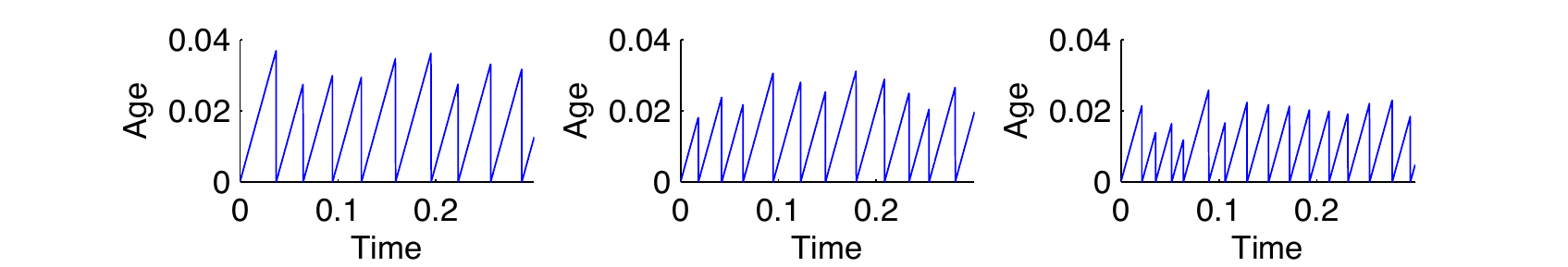}
   \caption{Simulation of the escape rate model (\ref{TEM}).
    The parameters of the simulation are $v_r=0.7$, $\mu=5$, $\sigma =0.1$, $\sigma =0.2$, $\sigma =0.3$.}\label{Fig6}
     \end{center}
\end{figure}

 As we pointed out in the introduction, the escape-rate models have been introduced in \cite{plesser} in order to arrive to
 more tractable models from mathematical point of view. It has been shown here that in the subthreshold regime for integrate
 and fire neurons, the diffusive noise
 can be replaced by a hazard noise (noisy threshold) described by a certain escape rate. More considerations about viable choices
of  age-dependent death rates as well as the derivation of the refractory-densities model that we will remind in the next section
can be found in \cite{gerstner}.

We present in  Fig. \ref{Fig6} a numerical simulation of the stochastic process defined by
 (\ref{TEM}) for different  age-dependent death rates. Note that the neuron never fires exactly at the same  age,
since its probability to fire (escape) is purely stochastic.

\section{The population density function (age-structure formalism)}

We can now introduce the AS model in the same way as it has been done in
\cite{gerstner}, \cite{perthame03} and \cite{perthame04}.

The model describes the evolution in time of the
population density function with respect to the  age of a neuron in the following way:
denoting by $n(t,a)$ the density of neurons at time $t$ at  age $a$, then the evolution of
$n$ is

\begin{equation}\label{age}
\frac{\partial }{\partial t}n(t,a)+ \overbrace{\frac{\partial }{\partial a}n(t,a)}^{\text{Drift part}}+
\overbrace{S(a)n(t,a)}^{\text{Spiking term}}=0.\\
\end{equation}%
In Fig. \ref{Age_Potential}, a schematic representation of the state space of the AS equation (\ref{age}) is presented.

\begin{figure}[t]
\begin{center}
   \includegraphics[width=6cm]{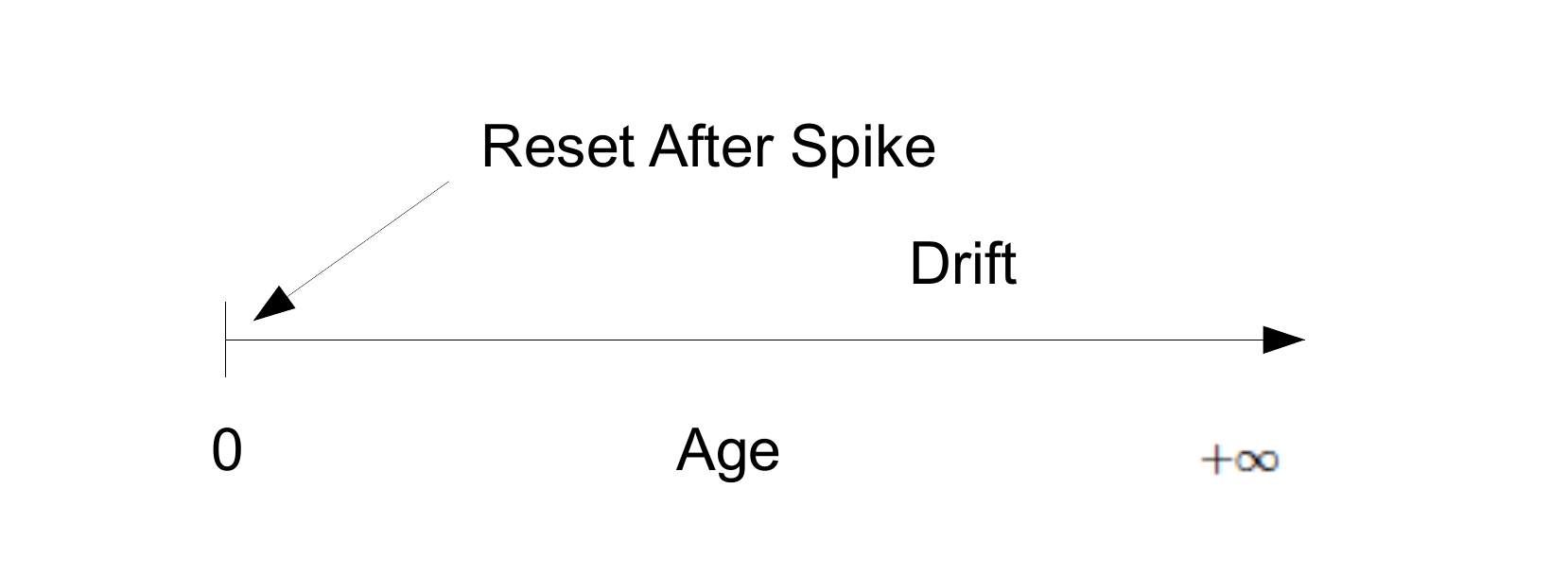}
   \caption{Schematic representation of the state space of  the AS equation (\ref{age}). }\label{Age_Potential}
      \end{center}
\end{figure}
Because once a neuron triggers a spike, its  age is reset to zero, we get the natural boundary condition
\begin{equation*}
\overbrace{n(t,0)=r(t)}^{\text{Reset}},\quad \forall t>0,
\end{equation*}
where $r(t)$ is the firing rate and is given by
\begin{equation}
\label{firing}
r(t)= \int_{0}^{+\infty}S(a)n(t,a) \,da,\quad  \forall t\geq 0.
\end{equation}
An initial distribution is assumed known:
\be\label{id}
n(0,a)=n_0(a), \quad \forall a>0.
\ee
In the above equations, $S(a)$ stands for the  age-dependent death rate given by (\ref{SF}).
Using the boundary condition and the expression of $r(t)$ given by (\ref{firing}), one can check easily the
conservation property of the equation (\ref{model}) by  integrating it on the interval $(0,\infty)$, so that if the initial
 condition  satisfies
\begin{equation}
\int  _{0}^{+\infty} n_0(a)\,da=1,
\end{equation}
the solution at any $t>0$ satisfies the normalization condition
\begin{equation}
\label{normalisation}
\int  _{0}^{+\infty} n(t,a)\,da=1.
\end{equation}

We present in  Fig. \ref{Fig7} a simulation of the problem (\ref{age})-(\ref{id}).
Again, we have made a comparison between the stochastic
process (blue curve) given by (\ref{SF}) and the evolution
in time of the density function (black curve) given by (\ref{age})-(\ref{id}). The simulation starts
with a Gaussian as initial condition (the first panel of Fig. \ref{Fig7}). Under the influence of the drift term,
the density function advances in  age, which is clearly seen in the  upper plots
of Fig. \ref{Fig7}. After the spiking process, the age of the neuron is reset to zero. The effect is well perceived
 in the lower panels of Fig. \ref{Fig7}. As expected from the model, the density function converges to an equilibrium density
 (see the last panel in Fig. \ref{Fig7}).
\begin{figure}[t]
\begin{center}
    \includegraphics[width=12.5cm]{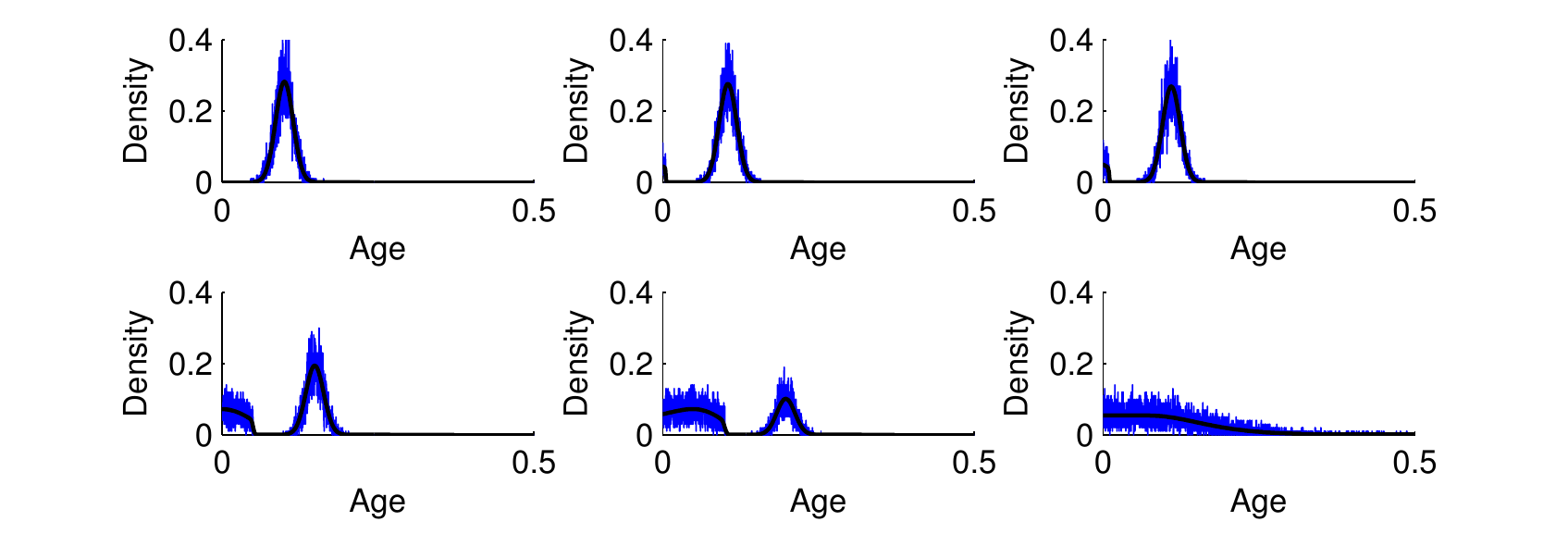}
   \caption{Simulation of the AS model (\ref{age})-(\ref{id}). Inhere, the black curve represents
   the simulation of the AS model and the blue curve  the simulation of the stochastic process (\ref{TEM}).
   A gaussian was taken as initial condition; the parameters of the simulation are: $v_r=0.7$, $\mu=5$, $\sigma =0.1$.
   The six plots in the figure show the evolution in time of the solution at $t=0$, $t=0.1$, $t=0.3$, $t=0.5$, $t=0.7$, $t=7$.}\label{Fig7}
      \end{center}
\end{figure}
The stationary state of the AS model can be easily computed; denoting by $r_\infty$ the stationary firing rate, we get:
\begin{equation*}
\displaystyle n_{\infty}(a)=r_{\infty}e^{-\int_0^a S( s)\,ds},
\end{equation*}%
and, if we take into account the normalization condition, we obtain the expression of the stationary firing rate

\begin{equation*}
r_{\infty}^{-1}=\int_0^{+\infty}e^{-\int_0^a S( s)\,ds}\,da.
\end{equation*}%
Since $\displaystyle e^{- \int_0^a S( s)\,ds}$ is the expression of the survivor function, its integral
over $(0,\infty)$ has the interpretation of the mean firing time, therefore the last
relation says nothing else than the fact that the stationary firing rate equals to
the inverse of the mean firing time.

\section{A theoretical link between the AS and FP problems}
In this section, we present our main result that introduces an analytical link between the two formalisms,
that states that there exists an integral transform that translates the solution to the  problem
(\ref{age})-(\ref{id}) into the solution to (\ref{model})-(\ref{IC}).

\begin{proposition}\label{ITransformation}
Let $p$ a solution to (\ref{model})-(\ref{IC}) and $n$ a solution to (\ref{age})-(\ref{id}),
and $p_0(v)$ and $n_0(a)$ two corresponding
initial distributions.
Then, if $p_0$ and $n_0$  satisfy
\begin{equation*}\label{ir}
p_0(v)= \int_{0}^{+\infty}\frac{q(a,v) }{\int_{-\infty}^{1} q(a,w) \,dw} n_0(a) \,da,
\end{equation*}%
the following relation holds true:
\begin{equation}
\label{C1706}
p(t,v)= \int_{0}^{+\infty}\frac{q (a,v) }{\int_{-\infty}^{1} q(a,w) \,dw} n(t,a) \,da.
\end{equation}%
Here, $q (a,v)$ is the solution to (\ref{FP})-(\ref{ic}).
\end{proposition}

\begin{remark}
The integral transform given by the equation (\ref{C1706}) can be interpreted with the help of probability theory.
Since the integral $\int_{-\infty}^{1}q(a,w) \,dw$ is the survivor function and $q (a,v)$ is the probability density for
a neuron to be at  age $a$ and at potential $v$, the kernel of the transform can be interpreted as the probability density
for a neuron to be at potential value $v$ given that it survived up to   age $a$.  The solution $n(t,a)$
denotes the density of population at time $t$ in state $a$, the integral over the whole possible states $a$ of the
kernel multiplied by the density $n$ gives indeed the density of population at time $t$ in the state $v$.
\end{remark}

\begin{remark}
In proposition \ref{ITransformation}, the integral transform is given in the sense of distributions, as we shall
define it bellow.
\end{remark}

Let us  show for the beginning that the integral in (\ref{C1706}) is well defined.\\
If we denote by $N(t,a)=\frac{n(t,a)}{\int_{-\infty}^1 q(a,v)\dd v}$, since
\bda
\frac{\partial}{\partial a}\frac{n(t,a)}{\int_{-\infty}^1 q(a,v)\dd v}&=&
\frac{\frac{\partial}{\partial a}n(t,a)}{\int_{-\infty}^1 q(a,v)\dd v}
+n(t,a)\left(\frac{\partial}{\partial a}\frac{1}{\int_{-\infty}^1 q(a,v)\dd v}\right)\\
&=&\frac{\frac{\partial}{\partial a}n(t,a)}{\int_{-\infty}^1 q(a,v)\dd v}
+S(a)\frac{n(t,a)}{\int_{-\infty}^1 q(a,v)\dd v},
\eda
where we have used the definition of $S(a)$,
one can see that $N$ is solution to the following system:
\be\label{AAS}
\left\{\begin{array}{ll}
\frac{\partial}{\partial t}N(t,a)+\frac{\partial}{\partial a}N(t,a)=0,\\
N(t,0)=r(t)=\int_0^\infty ISI(a)N(t,a)\dd a,\\
N(0,a)=\frac{n_0(a)}{{\int_{-\infty}^1 q(a,v)\dd v}}.
\end{array}\right.
\ee
The system above can be easily integrated
\bd
N(t,a)=\left\{\begin{array}{ll}
                     N_0(a-t),\quad a>t,\\
                     r(t-a),\quad t\geq a,
              \end{array}
       \right.
\ed
and the regularity of the solution is dictated by the regularity
of the initial condition. In particular, if $N_0\in L^1(0,\infty)$ then $ N(t,\cdot)\in L^1(0,\infty)$.
Also, choosing $n_0$ such that $N_0\in H^1(0,\infty)$, and as soon as  $N_0(0)=\displaystyle\int_0^\infty ISI(a)N_0(a)\dd a$,
we obtain that  $ N(t,\cdot)\in H^1(0,\infty)$.

On the other hand, $q(a,v)$ is the solution to (\ref{FP})-(\ref{ic}), which is a parabolic
equation with zero right hand side and homogeneous boundary conditions, therefore its exponential convergence towards
zero as $a\ra \infty$ for a.e. $v\in(-\infty,1]$ is immediate. We therefore can assert that, for every $\varepsilon$,
\bd
q(a,v)<c(v),\:\: a.e.\:\: v\in(-\infty,1], a>a_\varepsilon.
\ed
Since the product of a function from $L^\infty$ with an $L^1$ function is integrable,
we have then that the transform (\ref{C1706}) is well defined.

We also point out that, due to the large time behavior of $q$, we have that
 \bd
 \int_0^\infty S(a)\dd a=+\infty,
 \ed
condition which is known in age structured systems theory to imply that
$n(t,a)$ tends to zero as $a$ tends to infinity (which can be easily seen by simply integrating (\ref{age}) ).
We have chosen to work
on $[0,\infty)$ as the age interval, but one could have chosen, exactly as
in AS systems theory, to work on a finite interval $[0, A_{\max}]$, where $A_{max}$ is the maximal age that can be reached.
In this context, the condition that on the age interval, the integral of mortality rate to be infinity, has the biological
interpretation that the density of the population at ages bigger than maximal one is zero, therefore the meaning
of maximal age is exact. Of course, in our context, it would mean that all the neurons would have fired before
reaching this maximal value.
Also, let us notice that the system in $n$ has a classical solution on the defined domain; since the mortality rate
does not depend explicitly on $t$, the derivatives with respect to $t$ and $a$ exist in classical sense.\\
Before starting the proof, let us make some considerations over the  solutions
to the systems (\ref{model})--(\ref{IC}), respectively
(\ref{FP})-(\ref{ic}). We shall consider weak solutions to both systems in the sense introduced in \cite{carillo01}, namely:

\begin{definition}\label{wp}
A pair of nonnegative functions $(p,r)$ such that
 $$p\in L^2\left( 0,T ; L^2_+(-\infty,1)\right), r\in L^2_{+}(0,T)$$
is a solution to (\ref{model})--(\ref{IC}) if, for any test functions $\varphi(t,v) \in L^2([0,T]\times(-\infty,1])$ such that
$$\frac{\partial^2}{\partial v^2}\varphi(t,v), \frac{\partial}{\partial t}\varphi(t,v),
(\mu- v)\frac{\partial}{\partial v}\varphi(t,v)\in L^2((0,T)\times(-\infty,1)), \varphi(T,v)=0,$$ the
following relation takes place:
\bea\label{d1}
&&\nonumber\int_0^T\int_{-\infty}^1 p(t,v)\left[-\frac{\partial\varphi(t,v)}{\partial t}-(\mu-v)\frac{\partial\varphi(t,v)}{\partial v}
-\frac{\sigma^2}{2}\frac{\partial^2\varphi(t,v)}{\partial v^2}\right]\dd v\dd t\\
&&=\int_0^T r(t)[\varphi(t,v_r)-\varphi(t,1)]\dd t+\int_{-\infty}^1 p_0(v)\varphi(0,v)\dd v .
\eea
\end{definition}
As functions of the form $\Phi(v)\Psi(t)$ with $\Phi(v)\in L^2(-\infty,1)$ such that
$$(\mu-v)\Phi'(v), \Phi''(v)\in L^2(-\infty,1),$$
and $\Psi(t)\in L^2(0,T)$ with $\Psi'(t)\in L^2(0,T), \Psi(T)=0$ are a dense subset of the test functions
in definition \ref{wp}, we will restrict (\ref{d1}) to
\bea\label{d2}
&&\nonumber\int_0^T\int_{-\infty}^1 p(t,v)\left[-\Psi'(t)\Phi(v)-\Psi(t)(\mu-v)\Phi'(v)
-\frac{\sigma^2}{2}\Psi(t)\Phi''(v)\right]\dd v\dd t=\\
&&\int_0^T r(t)[\Phi(v_r)-\Phi(1)]\Psi(t)\dd t+\int_{-\infty}^1 p_0(v)\Phi(v)\Psi(0)\dd v,
\eea
which gives the expression of the distributional derivative with respect to $t$:
\bea\label{d3}
&&\frac{\partial}{\partial t}\int_{-\infty}^1 p(t,v)\Phi(v)\dd v\\
&&\nonumber=\int_{-\infty}^1p(t,v)\left[(\mu-v)\Phi'(v)
+\frac{\sigma^2}{2}\Phi''(v)\right]\dd v
+ r(t)[\Phi(v_r)-\Phi(1)],
\eea
again, for all the test functions $\Phi$ defined as above.\\

In the same way we will use a weak formulation for (\ref{FP}) as
\bea\label{wq}
&&\nonumber\int_0^\infty\int_{-\infty}^1 q(a,v)\left[-\chi'(a)\Phi(v)-\chi(a)(\mu-v)\Phi'(v)
-\chi(a)\frac{\sigma^2}{2}\Phi''(v)\right]\dd v\dd a=\\
&&\Phi(v_r)\chi(0)-\Phi(1)\int_0^\infty\chi(a)ISI(a) \dd a,
\eea
with   $\chi(a)\in L^2(0,\infty)$ and $\chi'(a)\in L^2(0,\infty)$ and
$\Phi(v)$ as in the previous definition.\\

We can now proceed with our proof.

\begin{proof}

Let us notice for the beginning that the AS system in $n(t,a)$ can be formulated equivalently as (\ref{AAS})
in terms of the new variable
\bd
N(t,a)=\frac{n(t,a)}{\int_{-\infty}^1q(a,v)\dd v}.
\ed

Moreover, using this notation, the integral transform reads:
\be\label{ADP}
p(t,v)=\int_0^\infty q(a,v)N(t,a)\dd a.
\ee
In order to show that the above formula defines a solution to the FP system,
let us apply the distributional derivative to (\ref{ADP}) and show that it satisfies (\ref{d3}):
\bd
\frac{\partial}{\partial t}\int_{-\infty}^1 p(t,v)\Phi(v)\dd v=
\int_{-\infty}^1 \int_0^\infty q(a,v)\frac{\partial}{\partial t}N(t,a)\Phi(v)\dd a\dd v.
\ed
Using the fact that $N(t,a)$ is solution to (\ref{AAS}), the last equation becomes:
\bd
\frac{\partial}{\partial t}\int_{-\infty}^1 p(t,v)\Phi(v)\dd v=
-\int_{-\infty}^1 \int_0^\infty q(a,v)\frac{\partial}{\partial a}N(t,a)\Phi(v)\dd a\dd v.
\ed
Let us turn now to the definition of the weak solution $q(a,v)$ given by (\ref{wq}). Noticing that, under proper
assumptions over the initial state $N_0$,
for each t arbitrary but fixed, $N(t,a)$ and $\frac{\partial}{\partial a}N(t,a)$ are in $L^2(0,\infty)$, we
can apply the definition for $q$ by choosing $\chi(a)$ as $ N(t,\cdot)$. Then we get that
\bda
&&-\int_{-\infty}^1 \int_0^\infty q(a,v)\frac{\partial}{\partial a}N(t,a)\Phi(v)\dd a\dd v\\
&&=\int_{-\infty}^1 \int_0^\infty N(t,a)q(a,v)\left[(\mu-v)\Phi'(v)
+\frac{\sigma^2}{2}\Phi''(v)\right]\dd a\dd v\\
&&+\Phi(v_r)N(t,0)-\Phi(1)\int_0^\infty N(t,a)ISI(a) \dd a.
\eda
Since $N(t,0)=r(t)=\displaystyle\int_0^\infty N(t,a)ISI(a)\dd a$, the last two terms in the expression above give
\bd
r(t)[\Phi(v_r)-\Phi(1)]
\ed
and therefore, we have obtained that
\bda
&&\frac{\partial}{\partial t}\int_{-\infty}^1 \int_0^\infty q(a,v)N(t,a)\Phi(v)\dd a\dd v\\
&&=\int_{-\infty}^1 \int_0^\infty q(a,v)N(t,a)\left[(\mu-v)\Phi'(v)
+\frac{\sigma^2}{2}\Phi''(v)\right]\dd a\dd v+r(t)[\Phi(v_r)-\Phi(1)],\\
\eda
which is exactly (\ref{d3}) for $p(t,v)$ given by (\ref{ADP}), which completes our proof.

\end{proof}

The integral transform gives a corresponding relation between the two stationary states.
Due to the form of the  age-dependent death-rate $S$,
the stationary density of the AS systems is completely determined by the solution to (\ref{FP})-(\ref{ic}):

\begin{equation*}
n_{\infty}(a)=\dfrac{ \int _{-\infty}^1q (a,v)\,dv}{\int_0^{+\infty}\int _{-\infty}^1q (a,v)\,dv\,da}.
\end{equation*}%
Using now the integral transform, we obtain that the stationary solution to (\ref{model})-(\ref{IC})
satisfies
\begin{equation*}
p_{\infty}(v)=\dfrac{ \int _0^{+\infty}q (a,v)\,da}{\int_0^{+\infty}\int _{-\infty}^{1}q (a,v)\,dv\,da}.
\end{equation*}%
Since we have the relation
\bd
r_{\infty}= \frac{1}{\int_0^{+\infty}\int _{-\infty}^{1}q (a,v)\,dv\,da},
\ed
one may check directly that the above formula  is indeed a solution to the stationary problem of the potential-structured
system by multiplying (\ref{FP}) by $r_\infty$ and integrating  it on $(0,\infty)$.

\section{Asymptotic behavior}
In the previous section, we have shown that there is an integral transform relating the solution
of the time elapsed model
and the solution of the FP equation. Our integral transform goes that way: for a density
in  age, $n(t,a)$, one can associate a corresponding solution $p(t,v)$ to the (\ref{model})-(\ref{IC}).
Let us define the operator

\begin{equation*}
\begin{array}{ccccc}
F& : & L^1_+(0,+\infty) & \to & L^1_+(-\infty,1)  \\
& & n_\infty & \mapsto & p_\infty \\
\end{array}
\end{equation*}%
defined by
\begin{equation*}
p_\infty(v)= \int_{0}^{+\infty}\frac{q(a,v) }{\int_{-\infty}^{1}q(a,w) \,dw} n_\infty(a) \,da,
\end{equation*}%
where $p_\infty$ and $n_\infty$ are the stationary solutions to the potential- respectively age- structured systems.
In the previous section, we have shown that, as soon as the initial conditions $p_0$ and $n_0$ are related by $F$,
the whole trajectories $p(\cdot,v)$ and $n(\cdot,a)$ are also related by $F$. Here we show that in the case the relation
between initial
conditions is not satisfied, $F$ transforms the known convergence  of $n$ to $n_\infty$ for $t \rightarrow \infty$ in the
convergence of $p$ to $p_\infty=F n_\infty$. This gives an additional way to study the behavior of $p$ for large time, already studied
in  \cite{carillo01}, \cite{carrillo02} by other means.

\begin{proposition}\label{SetA}
For all initial conditions $p_0$ belonging to $L^1_+(-\infty,1)$, the solution $p$ of the
potential structured problem (\ref{model})-(\ref{IC})
converges to  $Fn_\infty$ where $n_\infty$ is such that, for any initial condition, the solution $n$ to (\ref{age}), verifies
\begin{equation*}
 \lim\limits_{t \to +\infty}  \Vert n( t,\cdot) -n_\infty(\cdot)\Vert_{L^1_+(0,\infty)} =0.
\end{equation*}%
\end{proposition}

\begin{proof}

The model (\ref{age})-(\ref{id}) is a classical McKendrick-von Foerster model,
 well known in population dynamics, with the particularity that
the age specific mortality and fertility rates are the same. Then, defining the intrinsic reproduction number $\mathcal{R}_0$
as
\begin{equation*}
\mathcal{R}_0=\int_0^{+\infty}S(a)\exp(-\int_0^a S(a')da')\dd a,
\end{equation*}
by direct computations we get
\begin{equation*}
\mathcal{R}_0=1.
\end{equation*}
Then it is well known  that in this case, for $t\rightarrow \infty$, $n$ converges to $n_\infty$ satisfying
\begin{equation}\label{stationary}
 \frac{d }{d a}n_\infty(a)+
S(a)n_\infty(a)=0,  \quad a>0,
\end{equation}
with
\begin{equation}
n_\infty(0)=\int_0^\infty S(a)n_\infty(a)da,
\end{equation}
and
\begin{equation}
 \int_0^\infty n_\infty(a)da=1.
\end{equation}
Let us now define $\alpha(t,v)$ the solution to:

\begin{equation*}
\frac{\partial }{\partial t} \alpha (t,v)+\frac{\partial }{\partial v}[(\mu -v) \alpha (t,v)] -\frac{ \sigma ^2}{2}
  \frac{\partial^2 }{\partial v^2} \alpha (t,v)  = \delta(v-v_r)r(t),
\end{equation*}%
with boundary conditions  similar to (\ref{BC1}), (\ref{BC2}), $r(t)$ being the firing rate relative to (\ref{age})
and the initial condition is given by

\begin{equation*}
\alpha(0,\cdot)=\alpha_{0}=Fn_0.
\end{equation*}%
Then, thanks to Proposition \ref{ITransformation},
\begin{equation*}
\alpha(t,\cdot)=Fn(t,\cdot).
\end{equation*}
The transform $F$ being continuous from to $L^1_+(0,+\infty)$ to $L^1_+(-\infty,1)$, we get
\begin{equation*}
\alpha(t,v) \rightarrow p_\infty(v)=Fn_\infty \mbox{   in  } L^1(-\infty,1),
\end{equation*}
as $t\ra\infty$, where one can check as in Proposition \ref{ITransformation} that $p_\infty(v)$ satisfies
\be
\begin{array}{ll}
\frac{\partial }{\partial v}[(\mu -v) p_\infty (v)] -\frac{ \sigma ^2}{2}
  \frac{\partial^2 }{\partial v^2} p_\infty (v)  = \delta(v-v_r)r_\infty,
  \end{array}
\ee
with
\begin{equation*}
r_\infty=-\frac{\sigma^2}{2}\frac{\partial}{\partial v}p_\infty(1),\\
\end{equation*}
and boundary conditions
\begin{equation*}
\lim\limits_{v \to -\infty} (-\mu +v)p_\infty(v) +\frac{ \sigma ^2}{2}\frac{\partial }{\partial v}p_\infty(v) =0,\quad p_\infty(1)=0.
\end{equation*}
Now let us consider $\beta(t,v)=p(t,v)-\alpha(t,v)$. It satisfies
\begin{equation*}
\frac{\partial }{\partial t} \beta (t,v)+\frac{\partial }{\partial v}[(\mu -v) \beta (t,v)] -\frac{ \sigma ^2}{2}
  \frac{\partial^2 }{\partial v^2} \beta (t,v)  = 0,
  \end{equation*}
  with boundary conditions
  \begin{equation*}
\lim\limits_{v \to -\infty} (-\mu +v)\beta(t,v) +\frac{ \sigma ^2}{2}\frac{\partial }{\partial v}\beta(t,v) =0,\quad \beta(t,1)=0,
\end{equation*}
 and with the initial condition given by:
\begin{equation*}
\beta(0,v)=p_0(v)-Fn_0(a).
\end{equation*}
Using a change of variable similar to the one in \cite{carrillo02}, this equation is transformed in a heat equation on $(-\infty,0)$
with a zero Dirichlet condition in $0$. Then it is clear that $\beta(t,v)$ goes to $0$  as $t\rightarrow \infty$ in
$L^\infty(-\infty,1)\cap L^1(-\infty,1)$, which ends the proof.

\end{proof}

\section{Conclusions and perspectives}

It has been shown in \cite{plesser} that the integrate-and-fire model with stochastic input
can be mapped approximately onto an escape-rate model. Despite the fact that the two systems reproduce the same statistical activity,
no analytical connection between them has been given until now. This paper is intended as a first step in this direction.
We have proven here the existence of
an exact analytical transform of the solution to the AS system into the solution to the F-P system which is an equivalent
description of the NLIF model.
 Our finding highlights the theoretical relationships between the two stochastic processes and explain why the
statistical firing distributions across time are similar for both models, see the red dots in Fig. \ref{AP_T}.
To our knowledge, such a result has not been proven until now.

  As we have pointed it out in the introduction section, there are several advantages in using the AS formalism,
and the main reason is that it has been already well-studied by mathematicians throughout the past decades.
Another advantage in using the age structured model regards its numerical simulations, see Fig. \ref{AP_T}.
While the NLIF model requires the numerical implementation of the Euler-Maruyama scheme,
the escape model can be simulated via a Gillespie-like algorithm. However, the noisy threshold model is probably a little bit
more difficult to relate to the underlying biophysics of a cell. For this crucial point, one would prefer the use of NLIF model
where each parameter of the model can be easily measured by neuroscientists.

We have to stress though that the results obtained here have been proven in the case of a not-connected
neural population, which is a strong simplifying assumption.
The case we considered is known in the framework of renewal systems as a stationary process. A possible extension of the present transform for the case of interconnected neurons remains thus for us an open issue to be investigated. But the most important thing to be investigated, and which is currently in working progress,
is the existence of an inverse transform to the one introduced here.
Nevertheless, in the absence of such an inverse transform, we  proved here that
the set of solutions to the F-P system defined through our transform is an attractor
set of the solutions as $t$ tends to infinity.

The integral transform given in Proposition \ref{ITransformation} has a probability meaning and
this meaning can be interpreted using Bayes' rule. Indeed, the kernel of our transformation can be read out as $\mathcal{P}(v|a)$, the probability to find a neuron at potential $v$ knowing its age $a$. The most important feature
of this kernel is that it is time independent; the very nature of the age $a$
contains all the information about time that is needed to properly define the
integral transform. On contrary, to define an inverse transform, one faces the
problem of having a kernel that must depend on time. Transforming then the solution to
the FP system  into the solution to the AS system is therefore a little bit trickier.
Another important aspect about the nature of the AS formalism is that
the variable $a$ also entails information about the last firing moment. Indeed,
attributing an age to a neuron presupposes that the considered neuron has
already fired an action potential. From our perspective, the membrane potential
variable $v$ does not carry out such knowledge. There is therefore a hidden
information in the AS model that is not present in the FP approach.
It is therefore our belief that, in order to properly define an inverse transform, one
would be forced to assume that the FP initial density shares the information
about the last firing event. A compatibility condition on the initial data is then
required, we therefore believe that $p_0$ and $n_0$ should be related through the same
transform presented in this paper.

The easiest way to formalize all this is probably to write down the Bayes' rule
$$\mathcal{P}(a|v;t) = \mathcal{P}(v|a)*\mathcal{P}(a;t)/\mathcal{P}(v;t) = \mathcal{P}(v|a)*n(t,a)/p(t,v)$$
Note that the time-dependence of $\mathcal{P}(a|v;t)$ is inescapable.
Moreover, the full form of $\mathcal{P}(a|v;t)$ is not useful within the context of an inverse transformation,
since such an inverse transformation is trivial (reducing to $n(t,a)$ on both sides). Given these observations,
it might make sense to assume that the system is close to equilibrium. With this maximum-entropy-like assumption,
we can write a time-independent version of $\mathcal{P}(a|v;t)$:

$$\mathcal{P}(a|v) \approx \mathcal{P}(v|a)*n_{\infty}(a)/p_{\infty}(v)$$
 where $n_{\infty}(a)$ and $p_{\infty}(v)$ are the equilibrium distributions for $n(t,a)$ and $p(t,v)$, respectively.
 These equilibriums can be calculated.  With such reasoning, a perfect theoretical inverse  from the AS model to the FP
 representations can be found, and
an approximate inverse transformation can be constructed using the values of $n$ and $p$ at equilibrium, see Fig. \ref{AP_T}.

\begin{figure}[]
\begin{center}
    \includegraphics[width=6.5cm]{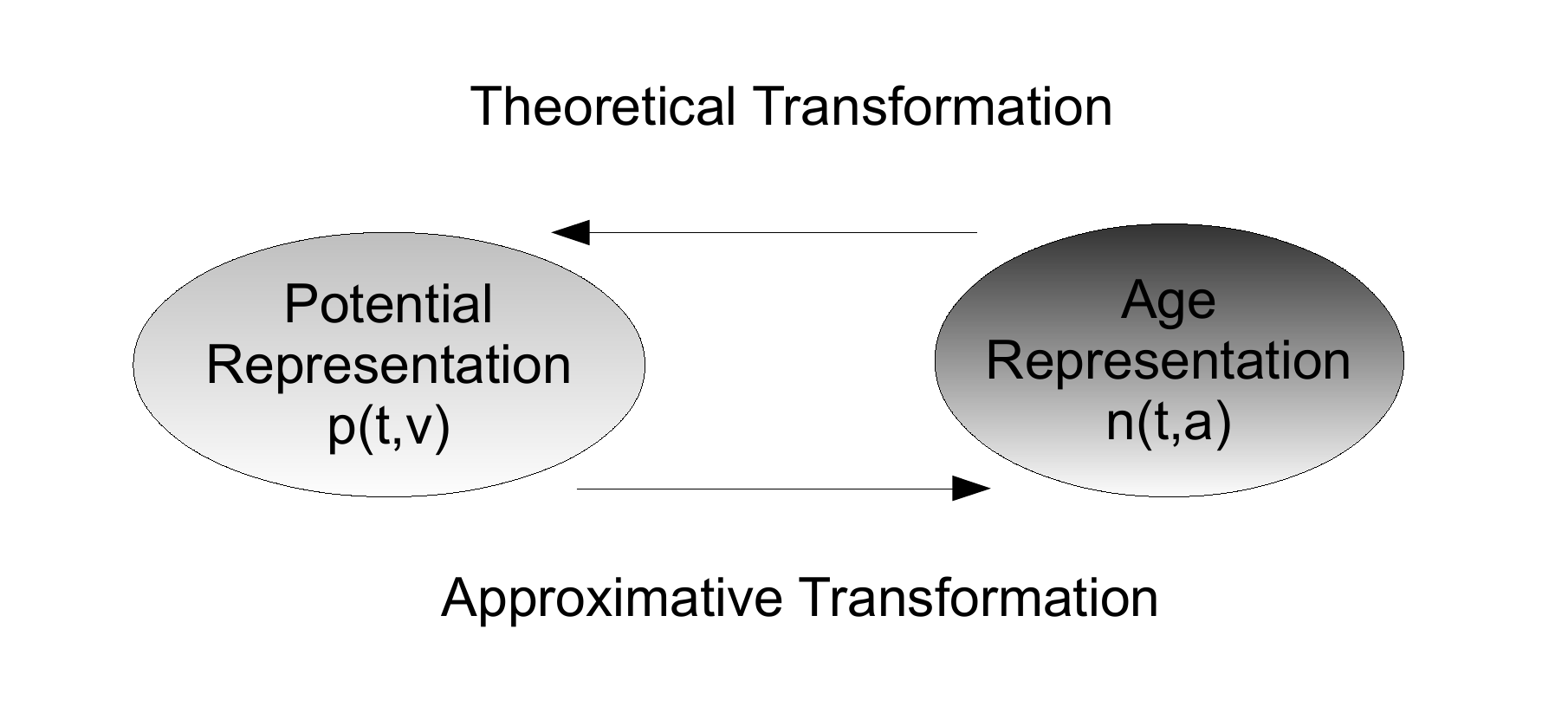}
     \includegraphics[width=9.5cm]{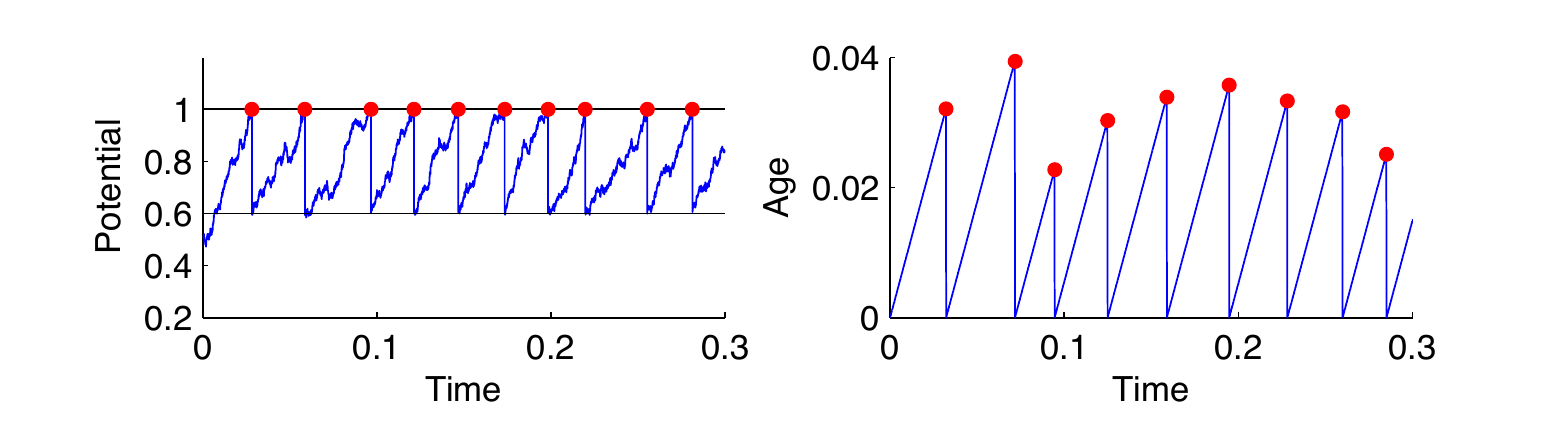}
   \caption{In the top panel, we show a schematic representation of the integral transformation between the two mathematical
    representations of the neural noise. In the bottom panel, we give two simulations illustrating the different mathematical
    treatments of noise in neuroscience context. The red dots indicates the firing time of the cell.
    The time distribution of those dots are
    similar in the two representations. The parameters of the simulation are $v_r=0.6$, $\mu=5$, $\sigma =0.2$.}\label{AP_T}
      \end{center}
\end{figure}

As we stressed in the introduction, the benefit of such
a representation would be the transfer of the analysis of special behaviors of the function $p(t,v)$ which is
the solution of a system that raises technical problems, to the study of the behavior of the AS
system, which is obviously simpler. The only quantity which will be significant then it will be the {\it age dependent death
rate} which will contain all the information needed to give insights of the behavior of the system.

\bibliographystyle{spmpsci}
\bibliography{art_fin}
\end{document}